\documentclass[authoryear,5p,10pt]{elsarticle}
\usepackage[T1]{fontenc}	

\usepackage{natbib}
\usepackage{pict2e}
\usepackage{ulem}
\usepackage{amscd}
\usepackage{amsmath}
\usepackage{amsthm}
\usepackage{amsfonts}
\usepackage{amssymb}
\usepackage{graphicx}
\usepackage{url}
\newtheorem{Theorem}{Theorem}
\newtheorem{Corollary}{Corollary}

\renewcommand{\emph}[1]{{\it#1}}
\newcommand{\proot}{a}
\newcommand{\finalversion}[1]{#1}

\newcommand{\balign}{\begin{align}}
\newcommand{\ealign}{\end{align}}

\newcommand{\0}{{\bm 0}}

\newcommand{\Fc}{{\cal F}}

\newcommand{\Order}{{\cal O}}

\newcommand{\Simplex}[1]{{\Delta^{#1}}}

\newcommand{\W}{\bm{W}}

\newcommand{\bdes}{\begin{description}}
\newcommand{\bd}{\begin{description}}
\newcommand{\bean}{\begin{eqnarray*}}
\newcommand{\benu}{\begin{enumerate}}
\newcommand{\ben}{\begin{enumerate}}

\newcommand{\bite}{\begin{itemize}}
\newcommand{\bi}{\begin{itemize}}

\newcommand{\bmu}{\begin{multline}}
\newcommand{\bm}[1]{{\bf #1}}
\newcommand{\boundary}{\partial}

\newcommand{\choo}[2]{{{#1}\choose{#2}}} 

\newcommand{\dspfrac}[2]{\frac{\displaystyle #1}{\displaystyle #2} }

\newcommand{\eban}{\begin{eqnarray*}}
\newcommand{\eba}{\begin{eqnarray}}
\newcommand{\eb}{\begin{equation}}
\newcommand{\edes}{\end{description}}

\newcommand{\ed}{\end{description}}
\newcommand{\eean}{\end{eqnarray*}}
\newcommand{\eea}{\end{eqnarray}}
\newcommand{\eenu}{\end{enumerate}}
\newcommand{\een}{\end{enumerate}}
\newcommand{\ee}{\end{equation}}
\newcommand{\eite}{\end{itemize}}
\newcommand{\ei}{\end{itemize}}
\newcommand{\emu}{\end{multline}}

\newcommand{\eqdef}{:=}

\newcommand{\ev}{\bm{e}}

\newcommand{\fb}{\ov{f}}

\newcommand{\figref}[1]{Fig. \ref{#1}}

\newcommand{\fv}{\bm{f}}

\newcommand{\g}{\bm{g}}
\newcommand{\gv}{\bm{g}}

\newcommand{\hide}[1]{}	

\newcommand{\intersect}{\cap}

\newcommand{\matrx}[1]{{\left[ \stackrel{}{#1}\right]}}

\newcommand{\notimplies}{\ \ \not \!\!\!\! \implies}

\newcommand{\ov}{\overline}

\newcommand{\sectref}[1]{section \ref{#1}}

\newcommand{\suchthat}{\colon}

\newcommand{\tp}{\top}	

\newcommand{\wb}{\ov{w}}

\newcommand{\xb}{\ov{x}}
\newcommand{\xh}{\hat{x}}

\newcommand{\xvb}{\ov{\x}}
\newcommand{\xvh}{\hat{\x}}
\newcommand{\xv}{\bm{x}}
\newcommand{\x}{\bm{x}}

\newcommand{\yb}{\ov{y}}

\newfont{\gilfont}{cmsy10 scaled\magstep0}
\newcommand{\Reals}{\mathbb{R}} 
\newcommand{\Complex}{\mathbb{C}} 

\begin{document}

\title{Proof of the Feldman-Karlin Conjecture on the \\Maximum Number of Equilibria in an \\ Evolutionary System}
\author{Lee Altenberg
\\{University of Hawai`i at Manoa}
\\ \url{altenber@hawaii.edu}
}

\begin{abstract}
Feldman and Karlin conjectured that the number of isolated fixed points for deterministic models of viability selection and recombination among $n$ possible haplotypes has an upper bound of $2^n - 1$.  Here a proof is provided.  The upper bound of $3^{n-1}$ obtained by Lyubich et al. (2001) using B\'{e}zout's Theorem (1779) is reduced here to $2^n$ through a change of representation that reduces the third-order polynomials to second order.  A further reduction to $2^n - 1$ is obtained using the homogeneous representation of the system, which yields always one solution `at infinity'.   While the original conjecture was made for systems of viability selection and recombination, the results here generalize to \finalversion{viability} selection with any arbitrary system of bi-parental transmission, which includes recombination and mutation as special cases. An example is constructed of a mutation-selection system that has $2^n - 1$ fixed points given any $n$, which shows that $2^n - 1$ is the sharpest possible upper bound that can be found for \finalversion{the} general space of selection and transmission coefficients. 
\ \\
To appear in \emph{Theoretical Population Biology}, \url{doi:10.1016/j.tpb.2010.02.007}
\end{abstract}
\maketitle

\section{Introduction}

In a tribute issue to the late Sam Karlin, \citet{Feldman:2009} recounts their early collaborations \citep{Feldman:and:Karlin:1968,Karlin:and:Feldman:1969,Karlin:and:Feldman:1970:Convergence,Karlin:and:Feldman:1970:Linkage}, and mentions a longstanding unsolved conjecture they proposed regarding the maximum number of isolated fixed points of the population genotype frequencies, under viability selection and recombination:
\begin{quote}
For the two-locus two-allele problem these considerations suggested a maximum of fifteen fixed points, and in our work with the symmetric viability model we demonstrated that fifteen was indeed realizable when recombination was present.  Amazingly, to this day, our conjecture that the maximum number of equilibria in any $n$-chromosome viability system and for any recombination arrangement is $2^n -1$ has not been proven, although there are no counterexamples. 
\end{quote}

Here I provide a proof, through a modification of the approach used by \citet{Lyubich:1992}.  The proof also generalizes the result to other genetic processes besides recombination, in fact, to any arbitrary biparental transmission system, as described below.  These results apply to systems of autosomal loci with discrete, \sloppy non-overlapping generations, random mating, and constant viability selection.

\section{The Model}
The dynamical system considered here is represented by the recursion:
\eb
\label{eq:Recursion}
x_i' = g_i(\x) \eqdef  \dspfrac{1}{\sum_{j,k=1}^n w_{jk} x_j x_k} \  \sum_{j,k=1}^n T_{ijk} w_{jk} x_j x_k 
\ee
where 
\bd
\item [$i \in \{1, \ldots, n\}$] indexes the $n$ possible gamete genotypes (i.e. \emph{haplotypes});
\item [$x_1, \ldots, x_i, \ldots x_n$] are the state variables, the frequencies of haplotypes in the population, so $x_i \geq 0$ and \sloppy $\sum_{i=1}^n x_i = 1$;
\item [$\x = \begin{pmatrix}x_1\\ \vdots\\ x_n \end{pmatrix}$], $\g(\x) = \begin{pmatrix}g_1(\x)\\ \vdots\\ g_n(\x) \end{pmatrix}$;
\item [$x_i'$] is the frequency of haplotype $i$ in the next generation;
\item [$w_{jk} = w_{kj} \geq 0$] is the fitness of the diploid genotype composed of haplotypes $j$ and $k$;
\item [$T_{ijk} = T_{ikj} \geq 0$] is the probability that diploid genotype $jk$ produces gamete genotype $i$, so 
\[
\sum_{i=1}^n T_{ijk} = 1 \ \mbox{for each } j, k.
\]
\ed
\emph{Perfect transmission} is said to occur if gamete genotypes are identical to one or the other of the parental haplotypes in equal proportions:
\[
T_{ijk} = \frac{1}{2} (\delta_{ij} + \delta_{ik}).
\]
where $\delta_{ij} = 1$ if $i=j$, $\delta_{ij}=0$ if $i\neq j$. 

This recursion defines the map
\[
\gv: \Delta^{n-1} \mapsto \Delta^{n-1}
\]
on the $n-1$ dimensional simplex, 
\[
\Delta^{n-1} = \{ \x \geq \0 \suchthat \sum_{i=1}^n x_i = 1 \} \subset \Reals^n.
\]
 
The entries of the $n$ by $n^2$ matrix $\matrx{T_{ijk}}$ of transmission probabilities are determined by the biological processes that occur during genetic transmission.   Recombination, mutation, gene conversion, segregation distortion, inversions, and all of their combinations, are all representable by sets of $T_{ijk}$.     An absence of any transforming processes, and no segregation distortion, produces perfect transmission.  In this absence of transformation processes, a multiple locus system is equivalent to a single locus system where each multilocus haplotype is formally a single allele.

Processes that are not covered by this framework include gene duplication and transposition, since the number of potential haplotypes becomes infinite.  Also, infinite alleles models \citep[Sec. 3.6]{Ewens:2004:MPG} are obviously not covered.  Deletions, however, within a fixed set of genes, are representable by $T_{ijk}$ with finite $n$.  

The set of all possible transmission matrices clearly includes many that do not correspond to any known biological processes, but this will be seen to be irrelevant, as the results apply to all transmission matrices.

\section{Prior Work}
Several previous studies prove results close to Feldman and Karlin's conjecture.  \citet{Moran:1963:Measurement}, at the tail end of his Conclusions, provides early insight into the question of how many isolated fixed points are possible in a system of selection and recombination, when he invokes  B\'{e}zout's Theorem \citeyearpar{Bezout:1779}: 
\begin{quote}
The above discussion also raises the interesting theoretical question of how many stationary points the adaptive topography for two unlinked loci can have. It is clear that in trivial cases such stationary points can fill up a linear or a real continuum. This occurs when the $w_{ij}$ are independent of $i$, or $j$, or both. We may, however, ask how many isolated stationary points are possible. The two equations typified by (15) are cubics in $P$ and $p$ and hence, by B\'{e}zout's theorem have at most 9 distinct isolated solutions, real or complex. 
\end{quote}

\citet{Tallis:1966}  shows that with selection and perfect transmission, there are a maximum of $2^n - 1$ isolated fixed points of $\gv$, each fixed point being a population containing haplotypes from one of the $2^n-1$ possible subsets of $n$ haplotypes, the empty set excluded. (Also see \citealt[Theorem 9.1.1]{Lyubich:1992}.) 

\citet{Lyubich:1992} shows that with arbitrary transmission but with no selection, $2^{n-1}$ is an upper bound on the number of isolated fixed points of $\gv$.  \citet{Lyubich:Kirzhner:and:Ryndin:2001} show that with arbitrary transmission \emph{and} selection, $3^{n-1} $ is an upper bound on the number of isolated fixed points of $\gv$.  

The approach taken in \citet{Lyubich:1992} and \citet{Lyubich:Kirzhner:and:Ryndin:2001} to obtain the upper bounds relies on B\'{e}zout's Theorem \citeyearpar{Bezout:1779}. 

As originally stated (translated from French):
\begin{Theorem}{\citet[47.  p. 24]{Bezout:1779E}}
The degree of the final equation resulting from an arbitrary number of complete equations containing the same number of unknowns and with arbitrary degrees is equal to the product of the exponents of the degrees of these equations.
\end{Theorem}
Here, the `final equation' is a reference to the elimination method for solving systems of polynomials.  An immediate application is:
\begin{Corollary}\citet[48.3.  p. 24]{Bezout:1779E}
The number of intersection points of three surfaces expressed by algebraic equations is not greater than the product of the three exponents of the degrees of these equations.
\end{Corollary}
The corollary is readily generalized to the intersection of arbitrary numbers of algebraic curves, to yield the version of B\'{e}zout's theorem utilized here, restated from \citet[p. 365-366]{Kollar:2008}:
\begin{Theorem}[B\'{e}zout's Theorem]\ \\
Let $f_1(\x), \ldots, f_n(\x)$ be $n$ polynomials in $n$ variables, and for each $i$ let $d_i$ be the degree of $f_i$.  Then either
\benu
\item the equation(s) $f_1(\x) = \cdots = f_n(\x) = 0$ have at most $d_1 d_2 \cdots d_n$ solutions; or
\item the $f_i$ vanish identically on an algebraic curve $C$, and so there is a continuous family of solutions.
\eenu
\end{Theorem}

\subsection{In the Absence of Selection}

B\'{e}zout's Theorem is applied as follows by \citet[pp. 294-296]{Lyubich:1992}.  In the absence of selection, the system of fixed points of \eqref{eq:Recursion} can be written as the zeros of a system of polynomials:
\eb
\label{eq:secondOrder}
0 = f_i(\x) \eqdef \sum_{j,k=1}^n T_{ijk}  x_j x_k  - x_i \ \  (i = 1, \ldots, n-1)
\ee
and
\eb
\label{eq:firstOrder}
0 = f_n(\x) \eqdef \sum_{i=1}^{n} x_i - 1 .
\ee
This gives $n-1$ equations of degree 2, and one equation of degree 1 ( \eqref{eq:secondOrder}  and \eqref{eq:firstOrder}, respectively).  Therefore, $\fv(\x) = \0$ can have at most $2^{n-1} * 1 = 2^{n-1}$ isolated solutions, which are the fixed points of $\gv$.

\citet[Theorem 8.1.4 and Corollary 8.1.7, pp. 295-296]{Lyubich:1992} invokes the Poincar\'{e}-Hopf index theorem to reduce the upper bound to $2^{n-1} - 1$  when $\gv$ maps all points on the boundary in the direction of the interior of the simplex:  i.e., 
\begin{align*}
\mbox{for }\x & \in \boundary \Delta^{n-1} \suchthat \\
& (1-\epsilon) \ \x + \epsilon \ \gv(\x) \in \Delta^{n-1} - \boundary \Delta^{n-1} \mbox{\ for small \ }\epsilon > 0, \notag
\end{align*}
or, more simply,
\eb\label{eq:BoundaryCondition}
\mbox{if } x_i = 0, \mbox{then } g_i(\x)  > 0.
\ee
\citet[p. 295]{Lyubich:1992} provides inferences between several related conditions.  Let 
$\Fc(\gv) = \{ \xvh \suchthat \gv(\xvh) = \xvh, \ \xvh \in \Delta^{n-1} \}$
represent the set of fixed points of $\gv$ on $\Simplex{n-1}$.  The conditions are:
\begin{align}
&\Fc(\gv) \intersect   \boundary \Delta^{n-1} = \emptyset;  \label{cond:BoundaryFree} \\
&\mbox{1 = Poincar\'{e}-Hopf index of } \gv \mbox{\ on } \boundary \Delta^{n-1}; \label{cond:PHopf}\\
&\matrx{T_{ijj}}_{i,j=1}^n \mbox{\ is an irreducible matrix}; \label{cond:Irreducible}\\
&T_{ijj}  > 0 \ \forall \ i, j \in \{1, \ldots, n \};  \label{cond:PositiveHomozygous}\\
&\gv(\x) \in \Simplex{n-1} -\boundary \Simplex{n-1}, \ \forall \x \in \Simplex{n-1}. \label{cond:Interior}
\end{align}
Lyubich \finalversion{proposes} the following chain of implications:
\begin{align*}
\eqref{cond:Interior} \iff  \eqref{cond:PositiveHomozygous} \implies \eqref{cond:Irreducible} \implies \eqref{cond:BoundaryFree} \implies \eqref{eq:BoundaryCondition}
\implies \eqref{cond:PHopf}
\end{align*}

Additional consideration, however, shows the correct chain of implications to be:
\begin{align*}
\eqref{cond:Interior} \iff  \eqref{cond:PositiveHomozygous} 
\left\{
\begin{array}{l}
\!\!\!\! \implies \eqref{cond:Irreducible} \implies \eqref{cond:BoundaryFree}\\ \\
\!\!\!\! \implies \eqref{eq:BoundaryCondition}
\left\{\begin{array}{l}
\!\!\!\! \implies \eqref{cond:BoundaryFree} \\
\!\!\!\! \implies \eqref{cond:PHopf}
\end{array}
\right.
\end{array}
\right.
\end{align*}
In particular, $\eqref{cond:BoundaryFree} \notimplies  \eqref{eq:BoundaryCondition} $.  Lyubich writes,  if ``the boundary of the simplex contains no fixed points [\eqref{cond:BoundaryFree}] then the vector field $Vx - x$ on the boundary is directed inside the simplex [\eqref{eq:BoundaryCondition}].''  Following is a (nonbiological) counterexample where \eqref{cond:BoundaryFree} and \eqref{cond:Irreducible} hold, but \eqref{eq:BoundaryCondition} does not (without claiming the theorems themselves to be invalid):  it is possible for the boundary to be without fixed points, and for the matrix $[T_{ijj}]$ to be irreducible (Lyubich uses the alternate term `indecomposable'), yet for $\gv$ to map boundary points toward the boundary and not the interior.  

Let $T$ simply rotate the index of each haplotype by -1:
\[
T_{ijk} = \left \{
\begin{array}{l l}
1 & \mbox{for \ \ }  i \!\!\!\!\! \pmod  n + 1 = j = k,  \\
1/2 & \mbox{for \ \ }  i \!\!\!\!\!\pmod  n + 1  = j \neq k,  \\
1/2 & \mbox{for \ \ }  i \!\!\!\!\!\pmod  n + 1   = k \neq j,  \\
0 & \mbox{otherwise}.
\end{array}
\right.
\]
An initial point $\x  =\ev_1 \eqdef (1, 0,  \cdots, 0)^\tp$ (\finalversion{$\scriptstyle \tp$}, the transpose) is taken by $\gv$ through a cycle of period $n$ through the vertices of the simplex, $\ev_1$, $\ev_n$, $\ev_{n-1}$, $\ldots$, $\ev_2$, $\ev_1$, $\ev_n, \ldots$.  

Every point on the boundary maps to a different point on the boundary.  
This can be seen because every boundary point must have some indices $i, i+1$ (modulo $n$) such that $x_{i}>0$ while $x_{i+1}=0$, but $T$ rotates the indices so that $g_{i}(\xv) = 0$ in the next generation, so no boundary point is fixed.   Moreover, for any point on the boundary with adjacent zeroes, i.e. $x_{i} = x_{i+1} = 0$, then $g_i(\x) = 0$, contrary to boundary condition \eqref{eq:BoundaryCondition}.  Here, boundary points of a sub-simplex map to other boundary points of that sub-simplex, so $(1-\epsilon) \x + \epsilon \, \gv(\x)$ remains on the boundary for all $0 \leq \epsilon \leq 1$.  Any fixed points must therefore be in the interior of the simplex, and by symmetry, this can only be the center, $\x = (1/n, \ldots, 1/n)^\tp$.  Thus, the exclusion of fixed points from the boundary does not imply that $\gv$ maps the boundary in the direction of the interior.  

What Lyubich is really after, however, is condition \eqref{cond:PHopf}.  So if one starts by assuming \eqref{eq:BoundaryCondition}, then \eqref{cond:PHopf} follows and the rest of Lyubich's  proof goes through: 
\benu
\item $\gv$ has a Poincar\'{e}-Hopf index of 1 on the boundary of the simplex  \eqref{cond:PHopf};
\item The index  must equal the sum of the indices of the fixed points in the interior of the simplex (see \citealt{Glass:1975:Topological} and \citealt{Glass:1975:Combinatorial} for an accessible explication of this approach);
\item The index of a non-degenerate fixed point must be $+1$ or $-1$;
\item Supposing that the maximum of $2^{n-1}$ fixed points is attained, then each must have multiplicity of 1, and is thus non-degenerate;
\item Having $2^{n-1}$ non-degenerate fixed points in the interior of the simplex, however, would produce a sum for their indices that is even, contrary to the index of $\gv$ under condition \eqref{eq:BoundaryCondition};
\item So there can be no more than $2^{n-1} - 1$ isolated fixed points of $\gv$ in the interior of $\Delta^{n-1}$ given \eqref{eq:BoundaryCondition}.  
\eenu

\subsection{In the Presence of Selection}

The above result is derived when selection is absent, and the only force acting is transmission of some arbitrary form.  When selection is included along with transmission, the system of fixed points becomes:
\begin{align}
\label{eq:thirdOrder}
0 = f_i(\x) \eqdef  \sum_{j,k=1}^n &T_{ijk} w_{jk} x_j x_k  - x_i \sum_{j,k=1}^n w_{jk} x_j x_k 
\\ &( i = 1, \ldots, n-1), \notag
\intertext{and}
0 &= f_n(\x) \eqdef  \sum_{i=1}^{n} x_i - 1  \notag.
\end{align}
The terms $x_i \sum_{j,k=1}^n w_{jk} x_j x_k$ in \eqref{eq:thirdOrder} are all of degree 3, so  \citet[eq. (33)]{Lyubich:Kirzhner:and:Ryndin:2001} apply B\'{e}zout's Theorem to obtain an upper bound of $3^{n-1}$ on the number of isolated fixed points.

The vastly larger value for the upper bound when selection is present, $3^{n-1}$, versus $2^{n-1}$ when selection is absent, seems counterintuitive, and
one suspects that $3^{n-1}$ can be sharpened.

\section{Results}
Closer examination of \eqref{eq:thirdOrder} finds that the equations are third order only due to the presence of a single common factor, the mean fitness $\wb \eqdef \sum_{j,k=1}^n w_{jk} x_j x_k$.  Such a structure suggests potentials from a change in the representation.  This is indeed the case, and the system can be made second order by introducing an additional variable $y$, and an additional equation that constrains $y$ to equal the mean fitness at equilibrium.  The result is as follows:

\begin{Theorem}[Generalized Feldman-Karlin Conjecture]
\label{Theorem:Main}
Consider the evolutionary system with viability selection, random mating, and arbitrary transmission of $n$ possible haplotypes, \finalversion{represented by the map $\gv$}:
\[
x_i' = g_i(x_1, \ldots, x_n) = \dspfrac{1}{\sum_{j,k=1}^n w_{jk} x_j x_k}  \ \sum_{j,k=1}^n T_{ijk} w_{jk} x_j x_k \ ,
\]
where $i = 1, \ldots, n$, $T_{ijk} \geq 0$, $\sum_{i=1}^n T_{ijk}=1$, $w_{ij} \geq 0$, and
\eb\label{eq:Condition}
x_i \geq 0, \ \sum_{i=1}^n x_i = 1. 
\ee
The number of isolated fixed points of $\gv$ is never greater than $2^n-1$.
\end{Theorem}

\begin{proof}
With the introduction of an additional variable $y$, the system of fixed points can be represented as $n+1$ equations in $n+1$ variables:
\begin{align}
0 = &\ f_i(\x, y) \eqdef  \!\! \sum_{j,k=1}^n T_{ijk} w_{jk} x_j x_k - x_i y\label{eq:n-1}\\  &(i = 1, \ldots, n-1), \notag \\
0 = &\ f_n(\x, y) \eqdef  \sum_{j=1}^{n} x_j - 1, \label{eq:n}\\
0 = &\ f_{n+1}(\x, y) \eqdef  \sum_{i,j=1}^n w_{ij} x_i x_j - y  \label{eq:n+1}. 
\end{align}
Since all fixed points are isolated by hypothesis, application of B\'{e}zout's Theorem to \eqref{eq:n-1}, \eqref{eq:n}, and \eqref{eq:n+1} gives an upper bound of $2^{n-1} * 1 * 2 = 2^n$ on the number of fixed points.
 
This comes within 1 of the value conjectured by Feldman and Karlin for the upper bound on the number of isolated fixed points of evolutionary systems \eqref{eq:Recursion}.  Their upper bound could be demonstrated if one could show that a \finalversion{solution} satisfying \eqref{eq:n-1}, \eqref{eq:n}, and \eqref{eq:n+1} always exists outside the simplex.  Complex-valued solutions, and \finalversion{solutions} `at infinity', qualify and must be included in the count of \finalversion{solutions} in B\'{e}zout's Theorem.   

A solution `at infinity' is formally accounted for by mapping the system to projective space; this is accomplished by a homogeneous form for the intersection equations \citep[pp. 16-21]{Shafarevich:1994:AsAuthor}.  A new set of variables is defined: 
\begin{align*}
&z, \\
&\xb_i \eqdef x_i \, z, \text{ and} \\
&\yb \eqdef y  \,  z .
\end{align*}
Substituting in \eqref{eq:n-1}, \eqref{eq:n}, and \eqref{eq:n+1}, one obtains $n+1$ equations in $n+2$ unknowns:
\begin{align}
0 = &\ f_i(\xv, y) \eqdef  \frac{1}{z^2} \sum_{j,k=1}^n T_{ijk} w_{jk} \xb_j \xb_k - \frac{1}{z^2} \xb_i \yb \label{eq:z1} \\ 
& (i = 1, \ldots, n-1),  \notag \\
0 = &\ f_n (\xv, y) \eqdef   \frac{1}{z }\sum_{j=1}^{n} \xb_j - 1, \label{eq:z2}\\
0 =  &\ f_{n+1} (\xv, y) \eqdef   \frac{1}{z^2} \sum_{i,j=1}^n w_{ij} \xb_i \xb_j - \yb / z  \label{eq:z3}. 
\end{align}
Multiplying both sides of \eqref{eq:z1} and \eqref{eq:z3} $z^2$, and \eqref{eq:z2} by $z$, one obtains  the homogeneous form:
\begin{align}
0 = &\ \fb_i(\xvb, \yb, z) \eqdef  \sum_{j,k=1}^n T_{ijk} w_{jk} \xb_j \xb_k - \xb_i \yb \label{eq:Hom_n-1}\\
 &(i = 1, \ldots, n-1),  \notag\\
0 = &\ \fb_n (\xvb, \yb, z) \eqdef  \sum_{j=1}^{n} \xb_j - z, \label{eq:Hom_n}\\
0 =  &\ \fb_{n+1} (\xvb, \yb, z) \eqdef  \sum_{i,j=1}^n w_{ij} \xb_i \xb_j - \yb \ z  \label{eq:Hom_n+1}. 
\end{align}
In the homogeneous representation, any non-trivial solution, $(\xvb, \yb, z) \neq (\0, 0, 0)$, gives as solutions all its scalar multiples $c \, (\xvb, \yb, z)$, $c \in \Complex$.  Hence all scalar multiples of a solution count as a single point (a point in projective space) when counting solutions.

The original system \eqref{eq:n-1} is obtained by setting $z=1$.  By setting $z = 0$, any nontrivial solution ($\xvb \neq \0$ or $\yb \neq 0$) corresponds to a solution `at infinity' for \eqref{eq:n-1}, \eqref{eq:n}, and \eqref{eq:n+1}, since it would give $\x = \xvb / 0$ and $y = \yb / 0$.  

Letting $z=0$ and $\xvb = \0$, we see that this is a solution:
\begin{align*}
\fb_i(\0, \yb, 0) = &\sum_{j,k=1}^n T_{ijk} w_{jk} 0*0 - 0 \yb = 0 \\ &(i = 1, \ldots, n-1), \\
\fb_n(\0, \yb, 0) =  &\sum_{j=1}^{n} 0 - 0  = 0,\\
\fb_{n+1}(\0, \yb, 0) = & \sum_{i,j=1}^n w_{ij} * 0 * 0 - \yb \ 0  = 0. 
\end{align*}
The variable $\yb$ is clearly unconstrained here, hence, the non-trivial solution is:
\[
(\xvb, \yb) = c \,  (\0, 1),  \ c \neq 0.
\]
This solution `at infinity' reduces by $1$ the number of possible finite-valued fixed points, leaving an upper bound of $2^n-1$ on the number of isolated fixed points in the simplex $\Delta^{n-1}$.
\end{proof}

An illustration of how the solution `at infinity' arises is given in \sectref{subsection:AtInfinity}.

\section{A Mutation-Selection System Bearing $2^n-1$ Isolated Fixed Points}

Here, I show that $2^n-1$ is the smallest possible upper bound on the number of isolated fixed points for the general space of selection-transmission systems \eqref{eq:Recursion}, by constructing  an example that attain this bound.  It is already known that the bound is attained by examples of systems with selection and perfect transmission, where one fixed point is located in the interior of the simplex and \finalversion{in} the interior of each sub-simplex on the boundary, including the vertices \citep{Tallis:1966}.  So what remains to be determined is whether this bound can also be attained under some form of imperfect transmission.  

Such systems can be produced through the homotopy continuation method \citep{Kotsffeas:2001,Li:2003:Numerical}, where one creates a continuous family of systems between a known system, and an unknown system with desired properties.   In this case, the homotopy will go from a known system of selection and perfect transmission that has $2^n-1$ fixed points, to systems with imperfect transmission, by perturbing the transmission probabilities.  Under proper conditions, the homotopy will produce paths from the fixed points of the known system to the fixed points of unknown systems.

The homotopy continuation method {\it per se} originated independently in the work of \citet{Garcia:and:Zangwill:1977:Finding}, \citet{Drexler:1977}, and \citet{Chow:Mallet-Paret:and:Yorke:1978} \citep{Li:1997}.  An essential part of this method can be found earlier in the `method of small parameters' of Karlin and McGregor \citeyearpar{Karlin:and:McGregor:1972:Equilibria,Karlin:and:McGregor:1972:Polymorphisms,Karlin:and:McGregor:1972:Application}:
\begin{quote}
Principle I [if a system of transformations acting on a certain set (in finite dimensional space) has a ``stable'' fixed point, then a slight perturbation of the system maintains a stable fixed point nearby] can be interpreted as a perturbation or continuity theorem. Starting with a given genetic system for which the nature of the equilibria can be fully delineated (for example, the classical multi-allelic viability model), it is desired to investigate a perturbed version of the model. \citeyearpar[p. 86]{Karlin:and:McGregor:1972:Equilibria}
\end{quote}
\citet[Theorem 4.4, p. 231]{Karlin:and:McGregor:1972:Polymorphisms} use the implicit function theorem to establish the existence, uniqueness, and nearness of fixed points under perturbed models.  The stability properties of the fixed points also are preserved by the additional assumption that all fixed points are hyperbolic.  Hyperbolicity is not needed, however \finalversion{--- only non-degeneracy ---} to establish the existence, uniqueness, and nearness of isolated fixed points \citep[p. 24]{Akin:1983}.

For the choice of the system to perturb, I use \eqref{eq:Recursion} with perfect transmission, which is formally a one-locus, multiple-allele system:
\eb
\label{eq:RecursionPerfTrans}
x_i' = g_i(\x) \eqdef  \dspfrac{1}{\sum_{j,k=1}^n w_{jk} x_j x_k} \  x_i \sum_{j=1}^n w_{ij} x_j 
\ee
\citet{Tallis:1966} gives two examples of \eqref{eq:RecursionPerfTrans}  with $2^n-1$ fixed points, a \emph{fully overdominant} and a \emph{fully underdominant} system.  The fully overdominant system is due to \citet{Wright:1949:Adaptation} (discussed in \citealt[p. 260]{Li:1955}), and has $w_{ii} = 1 - s_i$, $w_{ij} = 1$ for all $i \neq j$, where $0 < s_i \leq 1$.  The fully underdominant system has $w_{ii} = 1$, $w_{ij} = 1 - s$ for all $i \neq j$, where $0 < s_i \leq 1$.

The homotopy is created by perturbing perfect transmission \finalversion{using} uniform mutation:
\eb
\label{eq:SelMutExample}
T_{ijk} = (1-\mu) \frac{1}{2}(\delta_{ij} + \delta_{ik}) + \mu / n,
\ee
where $\mu$ is the mutation rate.  With $\mu > 0$, \eqref{eq:SelMutExample} satisfies conditions \eqref{cond:Interior} and \eqref{cond:PositiveHomozygous}, yielding conditions \eqref{eq:BoundaryCondition}, \eqref{cond:BoundaryFree} and \eqref{cond:PHopf}.  Thus none of the boundary fixed points under \eqref{eq:RecursionPerfTrans} can remain on the boundary for $\mu > 0$.

To keep all $2^n - 1$ fixed points in the simplex when $\mu > 0$, the fixed points of \eqref{eq:RecursionPerfTrans} on the boundary need to move inside the simplex.  Since $\gv$ points into the simplex for $\mu > 0$, toward these fixed points, they need to have stable manifolds \citep[pp. 193--239]{Wiggins:1990} that enter the simplex.  A situation that produces this is where all $n$ corners of $\Simplex{n-1}$ are stable sink nodes, and all $k$-allele ($k \in \{2, \ldots, n-1 \}$) polymorphic boundary fixed points are unstable source nodes with respect to the sub-simplex $\Simplex{k-1}$ for which they are interior points.  \citet[p. 578]{Kingman:1961:MPPG} showed that under \eqref{eq:RecursionPerfTrans}, the vertex equilibria are stable and all polymorphic equilibria unstable if $\W$ is positive definite.  This is the case with the fully underdominant system (Also see \citealt{Christiansen:1990}.).

Let the fitnesses be $w_{ij} = 1 + s \, \delta_{ij}$, where $s > 0$.  In \eqref{eq:RecursionPerfTrans}, the fixed points comprise: the vertices of the simplex; the centers of each $h-1$ dimensional sub-simplex on the boundary; and the center of the simplex.  Fixed points will be of the form
\eb
\label{eq:PerfectTransmissionExample}
( \underbrace{1/h, \ldots, 1/h}_{h },  \underbrace{0, \ldots, 0}_{n-h} )
\ee
in any of $\choo{n}{h}$ distinct permutations of the order, where $h$ varies from $1$ to $n$.  Hence the number of fixed points is $\sum_{h=1}^n \choo{n}{h} = 2^n - 1$.

Letting the mutation rate $\mu$ become positive, this system provides a constructive proof for the following:

\begin{Theorem}
The upper bound of $2^n - 1$ on the number of isolated fixed points for evolutionary systems \eqref{eq:Recursion} is sharp over the general space of selection and transmission coefficients.  In particular, for any $n$, a system with imperfect transmission can be constructed that attains the upper bound of $2^n-1$ fixed points in the simplex, using
$w_{ij} = 1 + s \, \delta_{ij}$, $s > 0$, and
\eb
T_{ijk} = (1-\mu) \frac{1}{2}(\delta_{ij} + \delta_{ik}) + \mu / n, \notag
\ee
for the ranges $0 < \mu < 1/2$ and $s \geq  \dspfrac{n}{1/\mu - 2},$
yielding:
\begin{align}
\label{eq:SelMutRecursion}
g_i(\x) = &\frac{1}{1 + s \sum_{j=1}^n  x_j^2} \\
&  \left[ (1-\mu) x_i + \frac{\mu}{n} + s(1-\mu) x_i^2  + \frac{s \mu}{n} \sum_{j=1}^n x_j^2
\right] \notag.
\end{align}
\end{Theorem}
\begin{proof}
For small $\mu$, \eqref{eq:SelMutRecursion} is a perturbation of the system with perfect transmission, so its fixed points will be close to the isolated and non-degenerate points \eqref{eq:PerfectTransmissionExample}.  By symmetry we can expect them to be of the form:
\eb
\label{eq:permutations}
\xvh = P( \ \underbrace{ \frac{1-\proot}{h}, \ldots,  \frac{1-\proot}{h} }_{h },  \ \underbrace{  \frac{\proot}{n-h}, \ldots,  \frac{\proot}{n-h}}_{n-h}  \ )
\ee
for $h = 1, \ldots, n-1$, where $P$ is any permutation of the order of the entries.  For $h=n$ or $h=0$,  $\xh_i = 1/n$ for all $i$.

To verify this form, and solve for $\proot$, we substitute into \eqref{eq:SelMutRecursion}:
\eban
\lefteqn{\frac{1-\proot}{h} (1 + s V) } \\
& \displaystyle = (1-\mu) \frac{1-\proot}{h} + \frac{\mu}{n} + s(1-\mu) \left(\frac{1-\proot}{h}\right)^2 + \frac{s \mu}{n}V&
\eean
where 
\begin{align*}
V \eqdef \sum_{j=1}^n x_j^2 &= h \frac{ (1-\proot)^2 }{ h^2 } + (n-h) \frac{ \proot^2 }{ (n-h)^2} \\
&= \frac{ (1-\proot)^2 }{ h} +  \frac{ \proot^2 }{ (n-h)},
\end{align*}
which yields three roots:
\begin{align*}
\proot &= 1 - h/n, \\
\intertext{and}
\proot =   & \frac{1}{2} + \mu  \left( \frac{1}{2} -\frac{ h}{n} \right) \\
& \pm 
\sqrt{ \left[\frac{1}{2} + \mu  \left( \frac{1}{2} -\frac{ h}{n} \right) \right]^2 - \frac{\mu(h+s)(n-h)}{n s}}.
\end{align*}
The first root yields the central equilibrium $x_i = 1/n$.  Substitution shows that the other two roots yield the two equilibrium forms in \eqref{eq:permutations}, where $h$ and $n-h$ are interchanged.  The term inside the radical must be non-negative for real solutions, and so imposes constraints on $\mu$ and $s$:
\begin{align}
\label{eq:SquareRootTerm}
\gamma(h) &\eqdef   \left[\frac{1}{2} + \mu  \left( \frac{1}{2} -\frac{ h}{n} \right) \right]^2 - \frac{\mu(h+s)(n-h)}{n s} \\
&= \left ( \frac{1}{2} -\frac{ h}{n} \right) ^2 \left(\mu + \frac{n}{ s}\right) \mu + \frac{1}{4} - \frac{\mu n }{4 s} - \frac{\mu}{2} 
 \geq 0. \notag
\end{align}

To obtain the conditions for $\gamma(h) \geq 0$, first the value of $h$ is found that minimizes $\gamma(h)$:
\begin{align*}
\frac{d}{d h} \gamma(h) =  -  \frac{2}{n}   \left( \frac{1}{2} -\frac{ h}{n} \right)  \left(\mu + \frac{n}{ s}\right) \mu
\end{align*}
So $\frac{d}{d h} \gamma(h)  = 0$ at $h = n/2$.  The second derivative,
\[
\frac{d^2}{d h^2} \gamma(h) = \frac{2}{n^2}  \left(\mu + \frac{n}{ s}\right) \mu,
\]
is positive, hence $\gamma(n/2)$ is a minimum:
\begin{align*}
\gamma(n/2) & =  \left ( \frac{1}{2} - \frac{1}{2} \right) ^2 \left(\mu + \frac{n}{ s}\right) \mu + \frac{1}{4} - \frac{\mu n }{4 s} - \frac{\mu}{2} \\
& = \frac{1 }{4} - \frac{\mu n}{ 4s} - \frac{\mu }{ 2}
\end{align*}
Therefore,
\begin{align*}
\gamma(n/2) \geq 0 \ &\iff \  \frac{1}{4} \geq \mu \left( \frac{n}{ 4s} + \frac{1 }{ 2} \right)
\\
& \iff \  \mu < 1/2 \mbox{ and } s \geq \frac{\mu \, n}{1-2 \mu}.
\end{align*}
If $n$ is odd, then the requirement is $\gamma(\frac{n+1}{2}) = \gamma(\frac{n-1}{2})  \geq 0$, but since  $ \gamma(n/2)$ is the minimum, the above constraints on $\mu$ and $s$ are sufficient to keep $\gamma(\frac{n+1}{2})$ non-negative.
In \eqref{eq:SquareRootTerm}, since $\mu(h+s)(n-h)/(ns) > 0$, then $\sqrt{\gamma(h)} < \frac{1}{2} + \mu  \left( \frac{1}{2} -\frac{ h}{n} \right)$ so $\proot > 0$, as required.

The number of fixed points is the number of permutations of \eqref{eq:permutations}, $2^n - 1$, the same as for \eqref{eq:PerfectTransmissionExample} with perfect transmission.

Therefore, $\gv$ in \eqref{eq:SelMutRecursion} is an example of a system with imperfect transmission that attains the upper bound of $2^n-1$ isolated fixed points, all in the interior of the simplex.  Thus, for the general space of fitness and transmission coefficients, the upper bound in Theorem \ref{Theorem:Main} is sharp and cannot be lowered.
\end{proof}

\begin{figure}
\vspace*{.05in} 
\centerline{\includegraphics[width=\columnwidth]{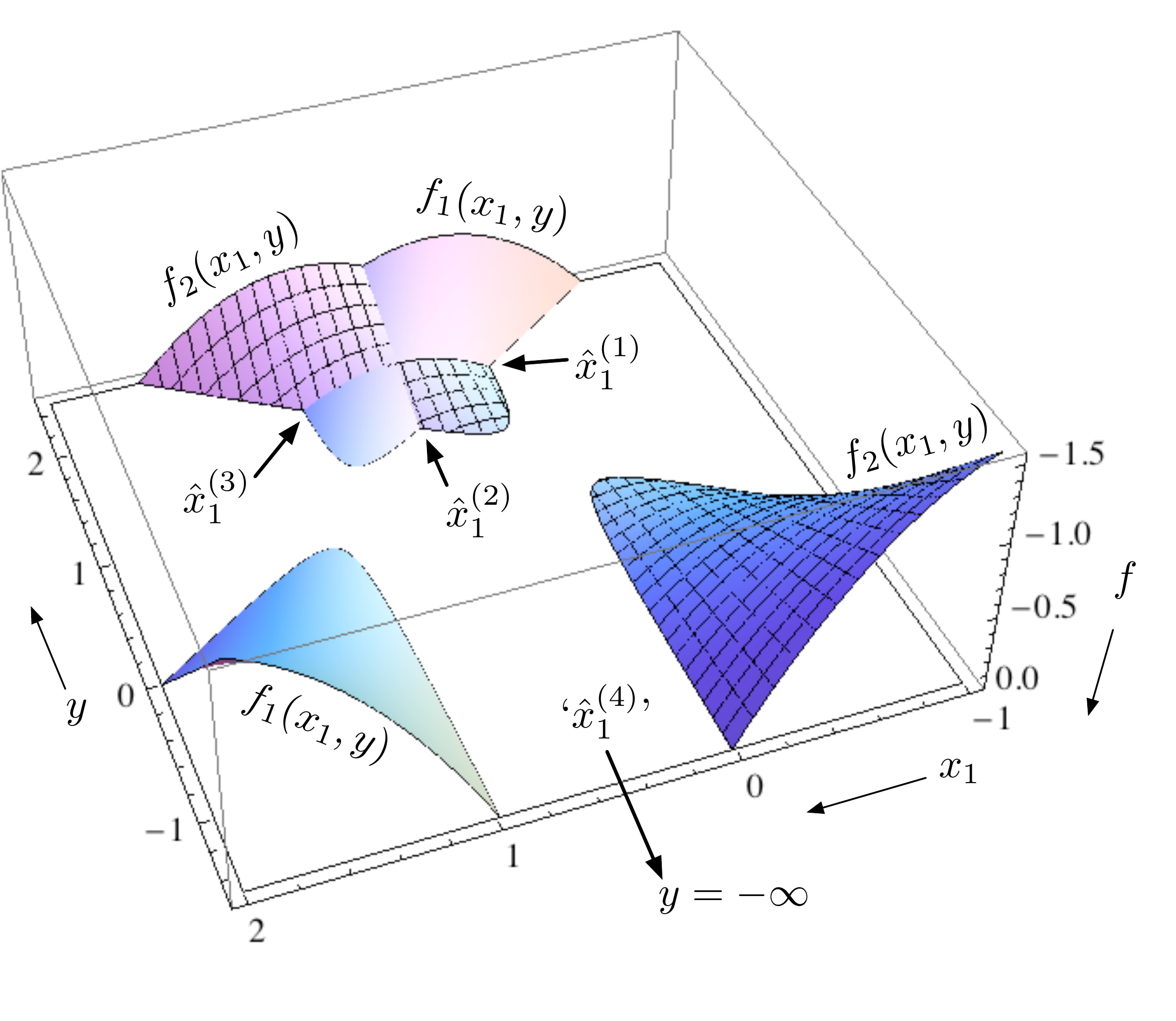}}
\caption{\label{fig:Example} A plot of the polynomials \eqref{eq:n-1} for the example \finalversion{\eqref{eq:SelMutRecursion}} with $n=2$, \finalversion{$s=1$}, $w_{11} = w_{22} = 2$, $w_{12} = w_{21} = 1$, \finalversion{and} $\mu = 0.05$.  The range is truncated at 0 to show the curves of solutions to $f_i(x_1, y)=0$.  Shown are the three points of intersection, where $f_1(\xh_1^{(k)}, y^{(k)}) = f_2(\xh_1^{(k)}, y^{(k)}) = 0$, for $k=1, 2, 3$.  At those three points, $\wb = y^{(k)}$.  Counting as the fourth intersection `at infinity', `$\xh_1^{(4)}$' indicates where the two curves approach parallel lines as $y \rightarrow - \infty$. }
\end{figure}

\subsection{The \finalversion{Solution} `At Infinity'}\label{subsection:AtInfinity}
This example system is also useful to understand the nature of the $2^n$-th solution `at infinity' in \finalversion{B\'{e}zout's} Theorem.  \figref{fig:Example} shows the polynomials \eqref{eq:n-1} for the two-allele case.  The two algebraic surfaces intersect each other and the $f_i = 0$ plane at $2^2 - 1 = 3$ points.  The absence of a fourth point is because the two algebraic surfaces approach parallel lines along the $f_i = 0$ plane as $y$ goes to $-\infty$.  These parallel lines are said to `meet at infinity' when the system is represented in projective space using the homogeneous form \eqref{eq:Hom_n+1}.

\section{Discussion}
Theorem \ref {Theorem:Main} proves the conjecture of Feldman and Karlin on the maximum number of isolated fixed points in a system of selection and recombination, and extends it to arbitrary transmission processes, of which recombination and mutation represent special cases.  This substantially sharpens the previous upper bound of $3^{n-1}$ \citep{Lyubich:Kirzhner:and:Ryndin:2001} on the number of isolated fixed points of \finalversion{an} evolutionary system with selection and arbitrary transmission. 

No attempt has been made to characterize the conditions on $T_{ijk}$ and $w_{jk}$ that produce only isolated and non-degenerate fixed points.  More on this issue can be found in \citet{Lyubich:Kirzhner:and:Ryndin:2001}.  One may mention, however, that such conditions are generic, in that for `almost all' sets \finalversion{of} $n$ algebraic hypersurfaces of degree $n$, the intersection consists of isolated non-degenerate fixed points (\citealt[p. 223]{Shafarevich:1994:AsAuthor}, \citealt{Garcia:and:Li:1980}); and sets of $T_{ijk}$ values that produce degenerate fixed points are nowhere dense \citep[Theorem 8.1.3]{Lyubich:1992} in the space of $T_{ijk}$ values.  Recombination-only systems without selection exhibit only continua of fixed points with no isolated fixed points, because arbitrary allele frequencies are all invariant, and only linkage disequilibria change in time.  Certain non-generic selection regimes also yield continua of fixed points.

The present paper does not touch at all upon the question of the stability of the equilibria.  Studies on the stability of equilibria include the papers cited previously, and, as a selected additional list,  \citet{Feldman:Franklin:and:Thomson:1974}, \citet{Karlin:1975}, \citet{Karlin:and:Liberman:1976}, \citet{Feldman:and:Liberman:1979}, \citet{Karlin:1979:Principles}, \citet{Karlin:1980},  \citet{Hastings:1981},  \citet{Hastings:1985}, \citet{Franklin:and:Feldman:2000}, and \citet{Puniyani:and:Feldman:2006}.  

It is notable that when selection is removed from the evolutionary system \eqref{eq:Recursion}, the upper bound on the number of isolated fixed points decreases by half, from $2^{n}-1$ to $2^{n-1}$ ($2^{n-1}-1$ when condition \eqref{eq:BoundaryCondition} holds).  In contrast, removing arbitrarily complicated imperfect transmission from \eqref{eq:Recursion} to constrain the systems to form \eqref{eq:RecursionPerfTrans}, where only selection acts, does not reduce the potential number of isolated fixed points at all.  

Removing selection means setting the $n \, (n+1)/2$ independent values of $w_{ij}$ to 1, whereas removing imperfect transmission means setting the $\finalversion{(n-1) \, n \, (n+1)/2}$ independent values of $T_{ijk}$ to $(\delta_{ij} + \delta_{ik})/2$.  Hence, paradoxically, when the $\Order(n^2)$ fitness coefficients are allowed to vary, it doubles the potential number of isolated fixed points, whereas when the $\Order(n^3)$ transmission probabilities are allowed to vary, the potential number of fixed points remains unchanged.  This reveals a structural difference between the roles of selection and transmission in the dynamics.

A homotopy of a continuous family of evolutionary systems can be defined between any two evolutionary systems \eqref{eq:Recursion}, in particular between ones with and without selection through the parameterizations:
\[w_{ij}(\alpha) =  (1-\alpha) + \alpha \ w_{ij}, \mbox{\ or \ } w_{ij}(\alpha) =  w_{ij}^ \alpha.
\]
So any system 
\[
g_i^{(\alpha)}(\x) \eqdef   \dspfrac{1}{\sum_{j,k=1}^n w_{jk}(\alpha) \ x_j \ x_k} \ \sum_{j,k=1}^n T_{ijk} \ w_{jk}(\alpha) \ x_j \ x_k
\]
with greater than $2^{n-1}$ fixed points will exhibit bifurcations of fixed points for some value(s) of $\alpha \in [0, 1]$.  

One can contemplate many other possible applications of the homotopy method to investigate the equilibrium behavior of these evolutionary systems.

\section{Acknowledgments}
The idea of creating a new variable to reduce the third order equations to second order came to me after seeing the use of substitutions in \citet{MathPages:Bezout}.

\end{document}